\documentclass[letterpaper, 10 pt, conference]{ieeeconf}  % Comment this line out if you need a4paper

\IEEEoverridecommandlockouts                             

\usepackage{graphics}
\usepackage[english]{babel}
\usepackage{epsfig} % for postscript graphics files
\usepackage{mathptmx} % assumes new font selection scheme installed
\usepackage{times} % assumes new font selection scheme installed
\usepackage{amsmath,mathtools} 
\allowdisplaybreaks
\usepackage{array}
\usepackage{mdwmath}
\usepackage{amssymb} 
\newtheorem{theorem}{Theorem}

\newtheorem{lemma}{Lemma}
\newtheorem{definition}{Definition}
\usepackage{subfigure}

\newcommand{\comment}[1]{{ }}
\title{\LARGE \bf
A Model-Free Optimal Control Method With Fixed Terminal States and Delay
}

\author{Mi~Zhou,
        Erik~Verriest,
        and~Chaouki~Abdallah
\thanks{Mi Zhou, Erik Verriest, and Chaouki Abdallah are with the School of Electrical and Computer Engineering, Georgia Institute of Technology, Atlanta, GA 30332. Emails:
        {\tt\small mzhou91@gatech.edu, erik.verriest@ece.gatech.edu, ccabdallah@gatech.edu.}}}

\begin{document}
\maketitle
\thispagestyle{empty}
\pagestyle{empty}
\begin{abstract}
Model-free algorithms are brought into the control system's research with the emergence of reinforcement learning algorithms.
However, there are two practical challenges of reinforcement learning-based methods.
First, learning by interacting with the environment is highly complex.
Second, constraints on the states (boundary conditions) require additional care since the state trajectory is implicitly defined from the inputs and system dynamics.
To address these problems, this paper proposes a new model-free algorithm based on basis functions, gradient estimation, and the Lagrange method.
The favorable performance of the proposed algorithm is shown using several examples under state-dependent switches and time delays.
\end{abstract}

\section{Introduction}
Optimal control is a widely-studied topic in robotics \cite{aeron}, aerospace engineering \cite{rocket}, and industrial applications such as fermentation \cite{fermentation} and inventory control \cite{inventory}.
Optimality conditions have been derived for systems with differentiable dynamics \cite{RichardVinter}, piecewise differentiable dynamics \cite{Mi-ACC2022,Mi-ACC2023} and time-delayed systems \cite{Bensoussan,BOCCIA}, as well as the combination of state-dependent switching and time delays \cite{Mi-TD}.
Direct and indirect numerical algorithms have also been proposed in \cite{ICLOCS2, Leek, GPOPS, nosnoc}. 

In addition to the above-mentioned numerical methods, differential dynamic programming (DDP) is a well-known and widely used method in applications.
Differential dynamic programming (DDP), first proposed by David Mayne in 1965, is an optimal control algorithm for trajectory optimization.
Instead of optimizing over the full state space, it optimizes around a nominal trajectory by taking local quadratic models of the dynamics and cost functions.
This allows one to find a local optimal solution to a nonlinear trajectory optimization problem.
DDP has wide applications in real robotic control.
However, it has two main challenges: 1) calculating the dynamics derivative during optimization is a computational bottleneck; 2) handling constraints on the states requires additional care \cite{unscentedDDP}.
Thus, in \cite{unscentedDDP}, the authors proposed a constrained unscented dynamic programming method for optimal control problems which eliminated the dynamics derivative computations and supported general state and input constraints using augmented Lagrangians.
The idea is to replace the gradient and Hessian calculations with approximations computed from a set of sample points.
In \cite{equalityDDP}, the authors proposed a differential dynamic programming method to solve optimal control with fixed terminal states.
DDP requires the smoothness of the system dynamics and stage cost function.
Many practical systems however have hybrid properties, for example, 1-D bouncing balls, bipedal robots walking, perching quadcopters, and so on.
Thus, the classic DDP-based methods fail to work for state-dependent switched systems and time-delayed systems.

With the development of model-free methods such as Koopman operator theory \cite{KoopmanOC} and reinforcement learning \cite{RLbook}, convenient APIs such as TensorFlow and PyTorch, and fast computing resources such as GPU and TPU, researchers now have access to model-free methods for control applications.
One difficulty of using reinforcement learning is its complexity and limitations in solving fixed terminal state optimal control problems.
In this paper, and as an extension of our work in \cite{Mi_SCC}, we address these problems by proposing a new model-free method.

A challenge for the optimal control of state-dependent switched systems is the non-differentiability of the co-state and state dynamics at the switching interface \cite{Mi-ACC2022,Mi-ACC2023}.
For time-delayed systems, the optimality conditions imply a strict coupling between the state and future co-state \cite{Mi-TD}.
Thus, we propose a method that does not require the knowledge of the system's dynamics and finds an optimal control policy based on gradient descent.
Our algorithm first parameterizes the control input with a set of basis functions,
then the gradient of the objective of the function with respect to these parameters is obtained by random sampling and the least-square method.
Finally, gradient descent is used to update the parameters.
When considering the constraint of the terminal state, we augment the equality constraint into the cost function and use a gradient accent method to update the Lagrangian multipliers.
Convergence of the proposed algorithm is analyzed and three examples are provided to show the performance of proposed algorithms.
A user-friendly toolbox based on MATLAB APP Designer is also open-sourced.

%%%%%%%%%%%%%%%%%%%%%%%%%%%%%%
This article is organized as follows:
In Section \ref{sec:problem}, we formulate our optimal control problem of state-dependent switched systems and time-delayed systems.
Then in Section \ref{sec:algo}, we introduce our proposed model-free algorithm.
Section \ref{sec:conv} presents the proof of convergence of the proposed method.
Section \ref{sec:GUI} presents the graphical user interface (GUI) of the proposed toolbox.
In Section \ref{sec:simulation}, we present three examples using the proposed algorithm.
Finally, we conclude in Section \ref{sec:conclusion} and suggest future improvements of the proposed algorithm.
%%%%%%%%%%%%%%%%%%%%%%%%%%%%%%%%%%
\section{Problem formulated} \label{sec:problem}
Our optimal control problem is defined as follows
\begin{align}
    \min J= \Psi(x(t_f)) +\int_{t_0}^{t_f} L(x,u) \mathrm{d}t \\
\nonumber    s.t., \dot x(t) = f(x(t),u(t)) \\
\nonumber    x(t_0)=x_0, \; x(t_f) = x_f.
\end{align}
When the system dynamics are region-dependent discontinuous, i.e., $f(x,u) = f_i(x,u), \forall x \in \mathcal{R}_i, i\in [1,2,\cdots ,k]$, $\mathcal{R}_i$ denotes region $i$, $\mathcal{R}_i \cap \mathcal{R}_j = \{x|g_{ij}(x)=0\}$, we call it state-dependent switched systems.

When the system dynamics have the form
\begin{align}
\dot x(t)=f(x(t), x(t-\tau), u(t)), t\in [t_0, t_f] \\
\nonumber x(t) = \phi(t), -\tau <t<t_0,
\end{align}
we call it time-delayed systems with constant delay in the states.
Without loss of generality, we let $t_0=0$ in this article.
It is well-known that under appropriate conditions, the solution of the defined optimal control problem may exist but may not necessarily be unique.
In our case, we aim to find at least a suboptimal control solution.
Furthermore, the following constraints are required
\begin{enumerate}
    \item The state $x(t)$ is at least an absolutely continuous function that satisfies the system dynamics.
    \item The system is controllable.
    \item The system dynamics is piecewise Lipschitz continuous.
\end{enumerate}
Based on these assumptions, we propose next a model-free algorithm to find a suboptimal control policy.
%%%%%%%%%%%%%%%%%%%%%%%%%%%%%%%%%%%%%%%%%%%%%%%
\section{Proposed algorithm}\label{sec:algo}
We first discretize the system into $N$ intervals using the first-order Euler method, i.e., $x_{t+1} = x_t+f(x_t, u_t)dt$.
The control input is parameterized by a linear combination of parameters $\theta_i$ and basis functions $\phi_i(t)$.
That is
\begin{align}
u(t) = \sum_{i=1}^m \theta_i \phi_i(t).
\end{align}

\subsection{Gradient estimation}
The least squares estimator is an unbiased estimator.
In this article, we obtain the cost function gradient with respect to the parameters by the least square method.
The gradient of the cost function using the finite difference method can be firstly written as
\begin{align*}
\frac{\partial J}{\partial \theta_i} = \frac{J(\theta_i+\epsilon e_i)-J(\theta_i-\epsilon e_i)}{2\epsilon},
\end{align*}
where $e_i$ are unit directional vectors; $\epsilon$ is a small constant denoting the magnitude of the perturbation of parameters $\theta_i$.
Write $\nabla J(\theta) = [\frac{\partial J}{\partial \theta_1}, \cdots, \frac{\partial J}{\partial \theta_m}]$ and $\Delta \theta = [\Delta \theta_1, \cdots, \Delta \theta_m]$.
The Taylor expansion of $J(\theta+\Delta \theta)$ is
\begin{align*}
 J(\theta+\Delta \theta) = J(\theta)+\nabla J(\theta)\Delta \theta +O((\Delta \theta)^2) .
\end{align*}
The gradient using the least square method can then be obtained as
\begin{align}
\nabla J_\theta \approx (\Delta \theta^\top \Delta \theta)^{-1} \Delta \theta \Delta J,
\end{align}
where $\Delta J=J(\theta+\Delta \theta)-J(\theta)$.

\subsection{Augmented Lagrangian method and Dual decomposition}
For the optimal control problem with fixed terminal states, when augmenting the terminal constraints, we have
\begin{align} \label{eqn:constrainedJ}
J=\Psi(x(t_f))+\int_0^{t_f} L(x,u)\mathrm{d}t + \mu(x(t_f)-x_f),
\end{align}
where $\mu$ is the Lagrange multiplier.

Using this Lagrange relaxation technique, the constrained optimization problem is transformed into an equivalent unconstrained problem.
The unconstrained problem is thus
\begin{align}
\max \min_{\theta}\;  J(\theta, \mu),
\end{align}
where $J$ is defined in Eqn. \eqref{eqn:constrainedJ}.
The goal is to find a saddle point $(\theta^*, \mu^*)$, which is a feasible solution.

Dual ascent is a classical method that finds a feasible point for equality-constrained optimization problems. 
The steps of dual ascent is
\begin{align} \label{eqn:dualUpdate}
\nonumber\theta_{n+1} &= \mathrm{argmin}_{\theta} \; J(\theta, \mu_n),  \; \mathrm{(minimization\; step)}\\
\mu_{n+1} &= \mu_n + \beta_n(x(t_f)-x_f), \; \mathrm{(update\; dual \; variable)}
\end{align}
where $\beta_k$ is the step size.

The convergence of dual ascent relies on the strict convexity of the objective function.
To ensure robustness, one can add another augmented quadratic term
\begin{align*}
J=\Psi(x(t_f))+\int_0^{t_f} L(x,u)\mathrm{d}t + \mu(x(t_f)-x_f)+\frac{\rho}{2}||x(t_f)-x_f||_2^2,
\end{align*}
where $\rho$ is called the penalty parameter.
Then we can use the above dual ascent iteration to obtain $\theta$ and $\mu$.
Adding the term can make the objective function differentiable under milder conditions and it converges faster compared to the one without this term.

Instead of using the iteration equation \eqref{eqn:dualUpdate}, we estimate the gradient of the objective function concerning the multiplier $\mu$.
Thus the updating equations for the parameters and the Lagrangian multiplier are
\begin{align*}
\theta_{n+1} =\Gamma_\theta \left[ \theta_{n} - \alpha_n \frac{\partial J(\theta_n, \mu_n)}{\partial \theta_n}\right ] \\
\mu_{n+1} = \mu_n + \beta_n \frac{\partial J(\theta_n, \mu_n)}{\partial \mu_n},
\end{align*}
where $\alpha_n$ is the step size, $\Gamma_{\ast}$ denotes a projection operator that keeps the control input in an admissible set.

\begin{lemma}
 Assume
\begin{align*}
\sum_{n=1}^\infty \alpha_n=\sum_{n=1}^\infty \beta_n=\infty, \;
\sum_{n=1}^\infty (\alpha_n^2+\beta_n^2)<\infty, \;
\frac{\beta_n}{\alpha_n} \rightarrow 0.
 \end{align*}
 The iterates $(\theta_n, \mu_n)$ converge to a fixed point {could be a local one} almost surely, which is a feasible solution.
\end{lemma}

\subsection{Framework of proposed algorithm}
Figure \ref{fig:algo} shows the framework of the proposed algorithm.
In this framework, we first parameterize the control input using traditional basis functions.
Then the gradient is estimated using the finite difference and least squares method.
The gradient is then used to update the parameters.
The algorithm stops when the $J$ converges, i.e., $|J_{n+1}-J_{n}|<tol$ where $tol$ is a small tolerance value.
\begin{figure}[!htp]
    \centering   \includegraphics[width=\linewidth]{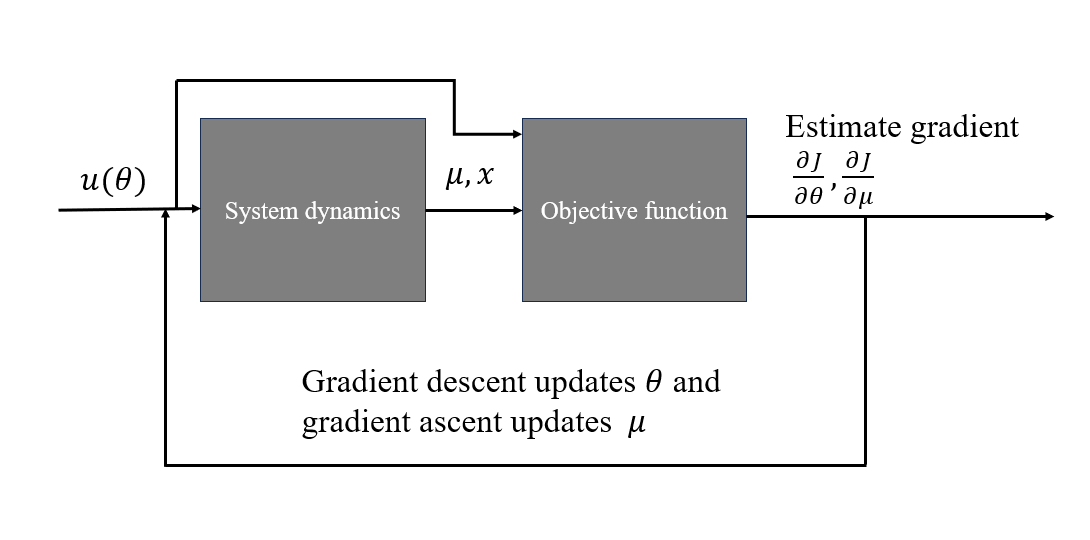}
    \caption{Framework of the proposed algorithm.}
    \label{fig:algo}
\end{figure}

%%%%%%%%%%%%%%%%%%%%%%%%%%%
\section{Convergence analysis}\label{sec:conv}
In this section, we prove the convergence of the proposed algorithm under some assumptions.
Firstly, we introduce the following definition of convergence:
\begin{definition}[Epi-convergence]
Let $U$ be a metric space and $J_n: U \rightarrow \mathbb{R}$ a real-valued function for each natural number. We say that the sequence $J_n$ epi-converges to a function $J: U \rightarrow \mathbb{R}$ if for each $u\in U$, 
\begin{align*}
\lim_{n\rightarrow \infty}\; \inf \;J_n(u_n) \geq J(u), \quad \forall u_n \rightarrow u ,
\end{align*}
and
\begin{align*}
\lim_{n\rightarrow \infty} \; \sup \; J_n(u_n) \leq J(u), \quad \forall u_n \rightarrow u .
\end{align*}
\end{definition}

Recall that $u(s,t) \sim \hat{u}(s,t)=\sum_{j=1}^m \theta_j(s)\phi_j(t)$ where $s$ denotes the index of a sequence of controllers approximating an optimal controller.
Here we will write $\hat{u}(t)$ as a truncated approximation of the optimal control $u(t)$.
\begin{lemma}(\cite{errbound,errBound2}) \label{lemma:u}
The theoretical error bound of $u(t)-\hat{u}(t)$ using basis with degree $m$ for 
$u\in W^{k,2}$ is
\begin{align*}
||u(t)-\hat{u}(t)||^2_{L^2(\Omega)} < \frac{C}{m^{k/n_u}} ,
\end{align*}    
where $C$ is some constant; $n_u$ is the dimension of the control input.
\end{lemma}
\begin{theorem}
Consider the optimal control problem with the proposed approximation $\hat{u}(t)$.
Assume $u^*(t)$ is the optimal solution of the optimal control problem.
For any $\epsilon>0$, there exists $m$ such that
\begin{align*}
|J(\hat{u}(t))-J(u^*(t))|<\epsilon.
\end{align*}
\end{theorem}
\begin{proof}
Based on the conclusion of Lemma \ref{lemma:u} and the definition of epi-convergence, we can show that
\begin{align*}
|J(\hat{u}(t))-J(u^*(t))| &< |\hat{u}(t)-u^*(t)| \\
&\leq |\hat{u}(t)-u(t)|+|u(t)-u^*(t)|    \\
& \leq \epsilon +\epsilon = 2\epsilon.
\end{align*}
\end{proof}
%%%%%%%%%%%%%%%%%%%%%%%%%%%%%%%%%
\section{MATLAB GUI toolbox} \label{sec:GUI}
Figure \ref{fig:gui} shows the design of the GUI where the initial state, terminal state, and terminal time are presented as a text block.
The system dynamics and control cost can be input from two functions.
When clicking the ``run'' button, there will be figures showing the states vs time, the control input vs time, and the convergence curve of the cost function vs iteration numbers.
There are three choices of basis functions: Chebyshev, Legendre, and Fourier.
The step size $\alpha$ can be tuned to compromise the speed and accuracy of convergence.
$\rho$ is the penalty for the augmented quadratic term.
The implementation and detailed instructions on the usage of the GUI can be found in the GitHub repository \footnote{https://github.com/jiegenghua/GOP}.

\begin{figure}[!htp]
    \centering    \includegraphics[width=\linewidth]{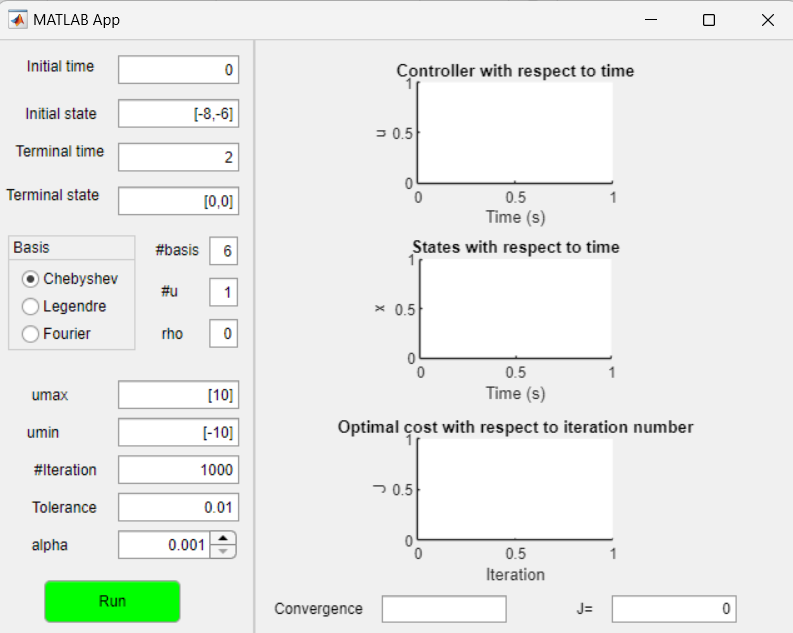}
    \caption{MATLAB GUI for the solver.}
    \label{fig:gui}
\end{figure}
%%%%%%%%%%%%%%%%%%%%%%%%%%%%%%%%%%%%%%%%%5
\section{Illustrated examples}\label{sec:simulation}
In this section, we demonstrate the performance of our algorithm in three examples with fixed terminal states and compare it to some existing methods.
\subsection{Example 1: a first order system}
Consider the following system:
\begin{align*}
\dot x = x+u,
\end{align*}
and the objective is to minimize the following cost function
\begin{align*}
\quad J=\int_0^1 (x^2(t)+u^2(t))\mathrm{d}t.
\end{align*}
The initial state is $x(0)=2$ and the terminal state is set up as $x(1)=4$. We let the stopping condition as $tol=0.01$ and sampling time $dt=0.01$ in all the following experiments.
The optimal solution in the last row is the solution using the ICLOCS2 toolbox \cite{ICLOCS2}.

We summarize the simulation results in Table \ref{tab:cartPole}.
\begin{table}[!htp]
    \centering
    \begin{tabular}{|c|c|c|} \hline
        Algorithm &  Optimal Cost ($J$) & Terminal state \\  \hline
         Chebyshev ($m=4$)& $ 8.1746\pm 0.0032$  & $3.9900\pm 0.0000$ \\ \hline
         Legendre ($m=6$)&  $ 8.1744\pm 0.0028$ & $3.9901\pm 0.0000$ \\ \hline
         Fourier ($m=4$, $\alpha=0.01$) & $ 8.1671 \pm 0.0014$ &  $3.9901\pm 0.0001$\\ \hline
        Optimal solution &  8.1445  & 3.9996 \\ \hline
    \end{tabular} 
    \caption{Performance compare for Example 1.}
    \label{tab:cartPole}
    \vspace{-20pt}
\end{table}
As we can see from the table, the solution obtained by the proposed algorithm has a small variance and gets close to the optimal solution.

Fig. \ref{fig:res_fo} (a)-(c) shows the state and control input using different basis functions, and the optimal cost obtained when using different number of basis functions. 
Fig. \ref{fig:heatmap_fo} shows the convergence of the coefficients under 10 experiments.
The last column is the parameters for the Lagrange multiplier. 
All the parameters are initialized as random variables at the beginning.
As we can see, all the parameters converge to a stable value.
\begin{figure*}[!htp]
\centering
\subfigure[]{
\includegraphics[width=0.25\textwidth]{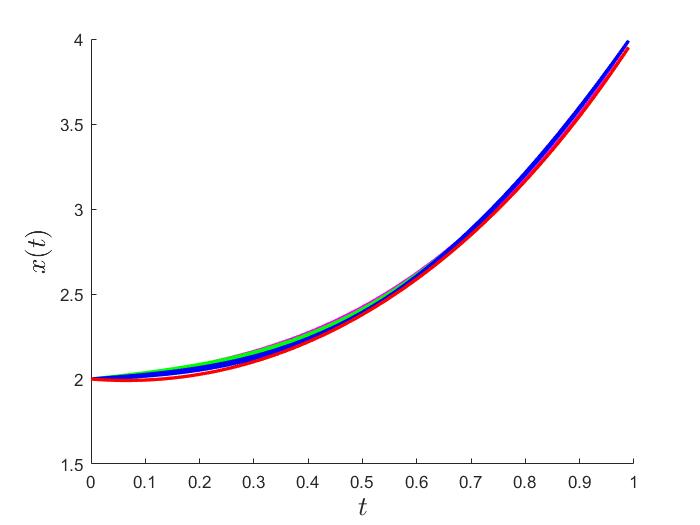}
}
\subfigure[]{
\includegraphics[width=0.25\textwidth]{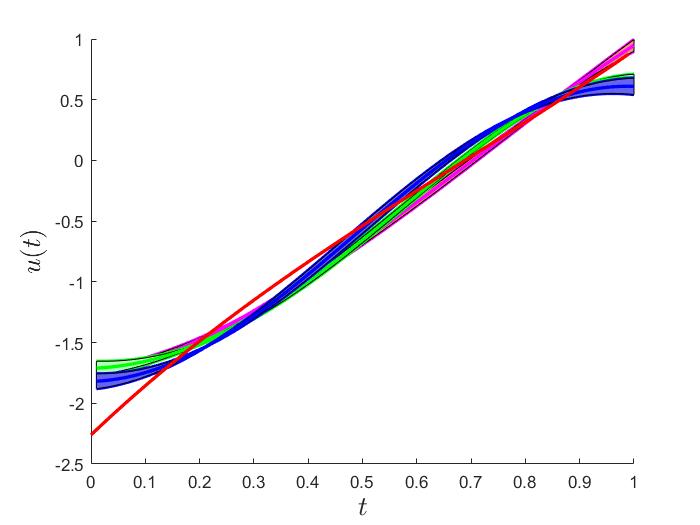}
}
\subfigure[]{
\includegraphics[width=0.25\textwidth]{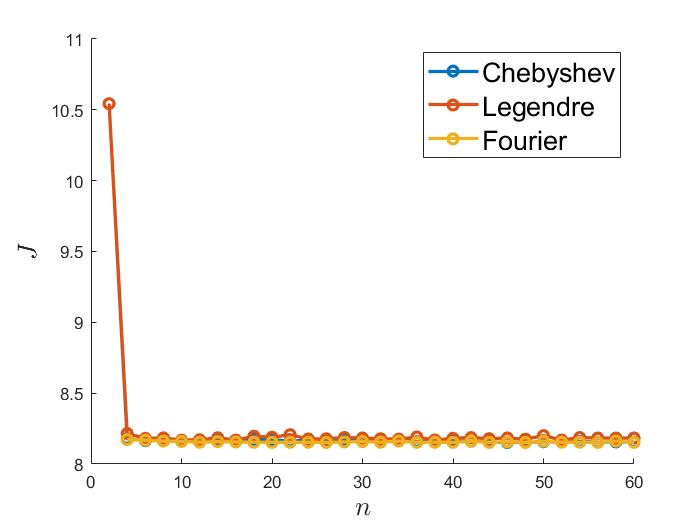}
}
\caption{\textbf{Example 1:} (a) state $x(t)$ under different basis functions (magenta: Chebyshev; green: Legendre; blue: Fourier); (b) control input $u(t)$ under different basis functions (magenta: Chebyshev; green: Legendre; blue: Fourier); (c) cost $J$ with respect to number of basis functions used.}
\label{fig:res_fo}
\end{figure*}

\begin{figure*}[!htp]
\centering
\subfigure[]{
\includegraphics[width=0.25\textwidth]{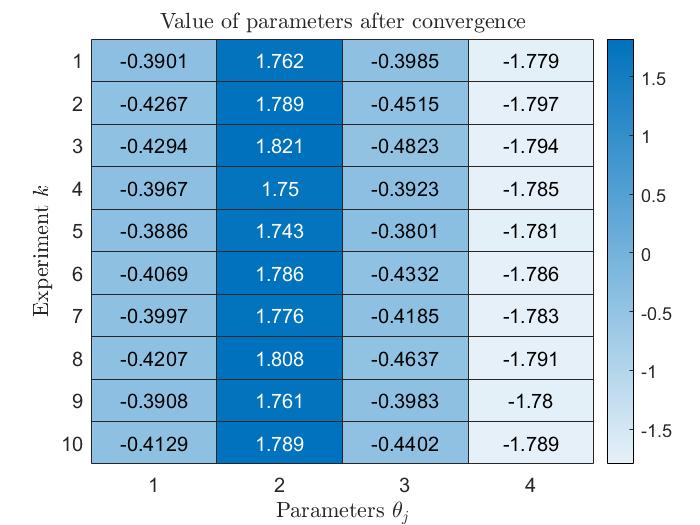}
}
\subfigure[]{
\includegraphics[width=0.25\textwidth]{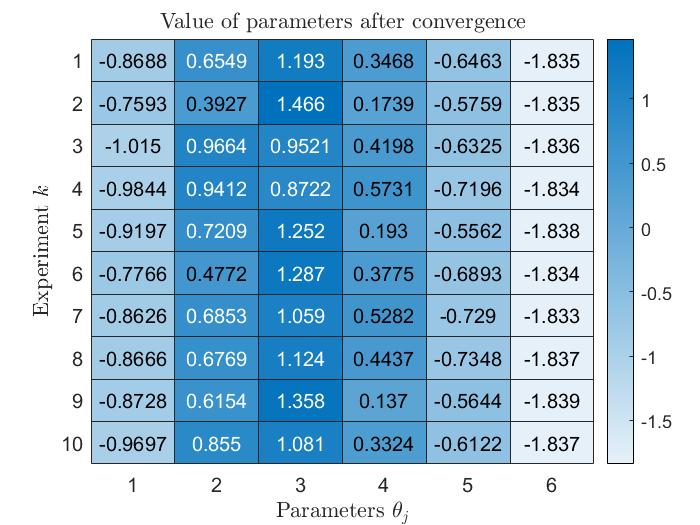}
}
\subfigure[]{
\includegraphics[width=0.25\textwidth]{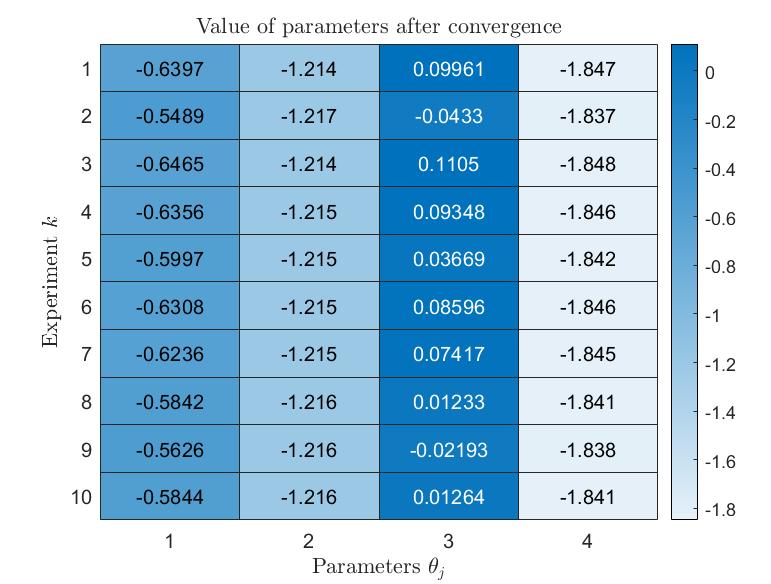}
}
\caption{\textbf{Example 1:} Heatmap of the parameters after convergence using different basis: (a)Chebyshev ($m=4$, $\alpha=0.01$) (b) Legendre ($m=6$, $\alpha=0.01$) (c) Fourier ($m=4$, $\alpha=0.01$).}
\label{fig:heatmap_fo}
\end{figure*}

\subsection{Example 2: State-dependent switched systems}
We consider the same system as in \cite{Mi_SCC} which has four regional dynamics:
\begin{align*}
\dot x = A_q x +B u,
\end{align*}
with
\begin{align*}
&A_1 = \begin{bmatrix}
-1 & 2 \\
-2 & -1 \\
\end{bmatrix},\; A_2 = \begin{bmatrix}
-1 & -2 \\
1 & -0.5 \\
\end{bmatrix},\; A_3=\begin{bmatrix}
-0.5 & -5 \\
1 & -0.5 \\
\end{bmatrix}, \\&
A_4 = \begin{bmatrix}
-1 & 0 \\
2 & -1 \\
\end{bmatrix}, \; B=\begin{bmatrix}
1 \\
1
\end{bmatrix}.    
\end{align*}
The stage cost function is $L = \frac{1}{2}(x^\top x+ u^2)$ for all regions and the fixed terminal time is $t_f = 2$. 
The initial state is $x_0 = (-8,-6)^\top$.
The control input $u(t)$ is constrained in $[-10,10]$.
The switching interfaces are $m_{12} = x_2+5 = 0$, $m_{13} = x_1+5=0$, $m_{23}=-m_{32}=x_1-x_2=0$, $m_{24}=-m_{42}=x_1+2=0$, $m_{34}=-m_{43}=x_2+2=0$.
Moreover, we fix the terminal state as $[0,0]^\top$.

We compare our result with the result obtained in \cite{graphsearch}.
For this system, we observed a slightly high variance of the control and state at the switching interface $x_2=-5$ and $x_2=-2$ as shown in Fig. \ref{fig:multiregionSys}(b).
This is also shown in Fig. \ref{fig:multiregionBasis} where at each experiment, the converged parameters are slightly different.
This is due to the high sensitivity at the switching interfaces when doing gradient estimation.
\begin{table}[!htp]
    \centering
    \begin{tabular}{|c|c|c|} \hline
        Algorithm &  Optimal Cost & Terminal state \\  \hline
         Chebyshev ($m= 28$)& $ 21.2891\pm 0.7610$  &  $ \begin{bmatrix}
           -0.0572\\ -0.0100
         \end{bmatrix}\pm \begin{bmatrix}
          0.1554\\ 0.2644    
         \end{bmatrix}$ \\ \hline
         Legendre ($m= 28$)&  $21.0954 \pm 0.3323$ &  $\begin{bmatrix}
           -0.0042\\ -0.0283  \end{bmatrix} \pm \begin{bmatrix}
             0.0442\\ 0.0827 
           \end{bmatrix}$ \\ \hline
         Fourier ($m=27$) & $ 21.4617 \pm 0.2486$ &  $\begin{bmatrix}-0.0400\\ -0.0814
\end{bmatrix} \pm \begin{bmatrix}0.1042\\0.1836 \end{bmatrix}$ \\ \hline
        Benchmark \cite{graphsearch} &  21.6980  &  $\begin{bmatrix}
         0\\0   
        \end{bmatrix}$ \\ \hline
    \end{tabular} 
    \caption{Performance compare for Example 2.}
    \label{tab:multiregion}
    \vspace{-20pt}
\end{table}

\begin{figure*}[!htp]
\centering
\subfigure[]{
\includegraphics[width=0.35\textwidth]{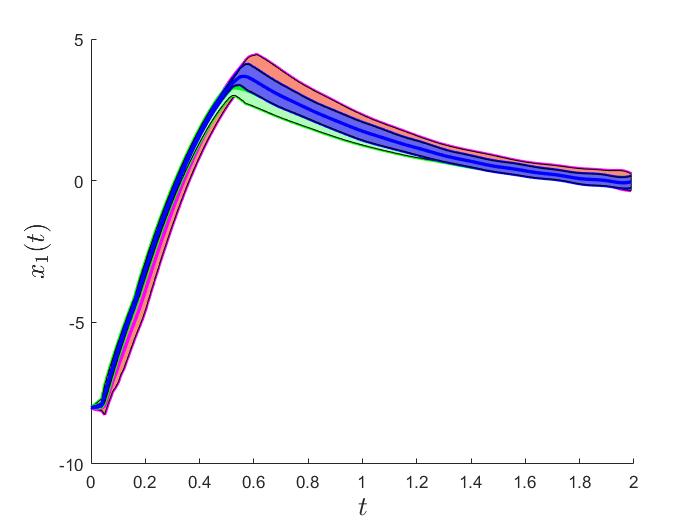}
}
\subfigure[]{
\includegraphics[width=0.35\textwidth]{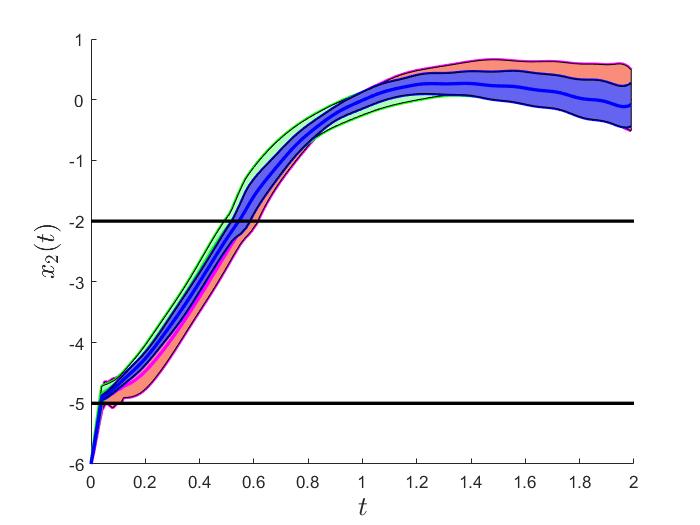}
}
\subfigure[]{
\includegraphics[width=0.35\textwidth]{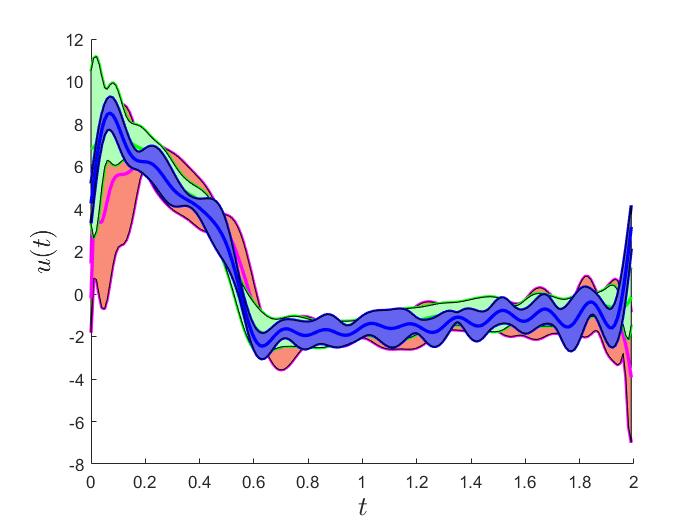}
}
\subfigure[]{
\includegraphics[width=0.35\textwidth]{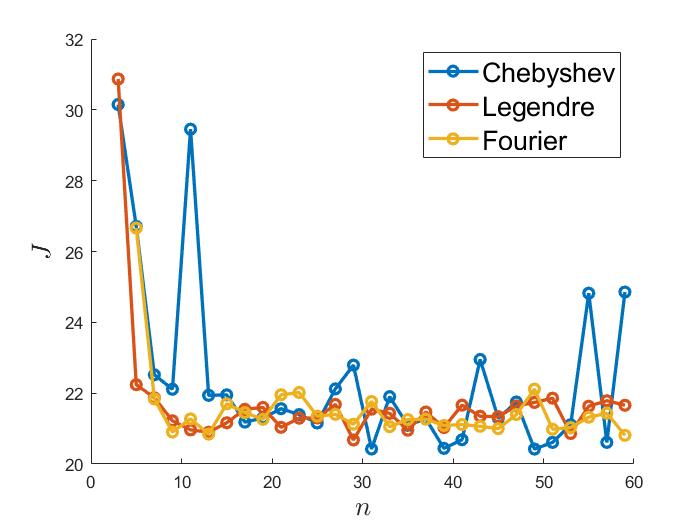}
}
\caption{\textbf{Example 2:} (a) state $x_1(t)$ under different basis functions (magenta: Chebyshev; green: Legendre; blue: Fourier); (b) state $x_2(t)$; (c) control input $u(t)$ under different basis functions (magenta: Chebyshev; green: Legendre; blue: Fourier); (d) cost $J$ with respect to number of basis functions used.}
\label{fig:multiregionSys}
\end{figure*}
\begin{figure*}[!htp]
\centering
\subfigure[]{
\includegraphics[width=0.25\textwidth]{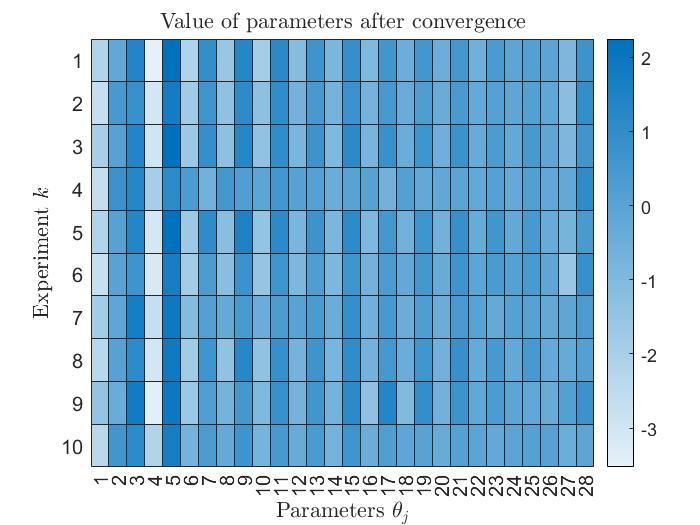}
}
\subfigure[]{
\includegraphics[width=0.25\textwidth]{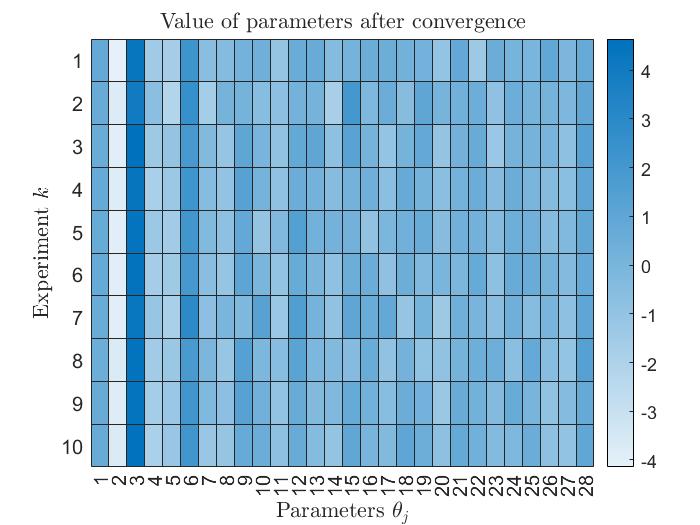}
}
\subfigure[]{
\includegraphics[width=0.25\textwidth]{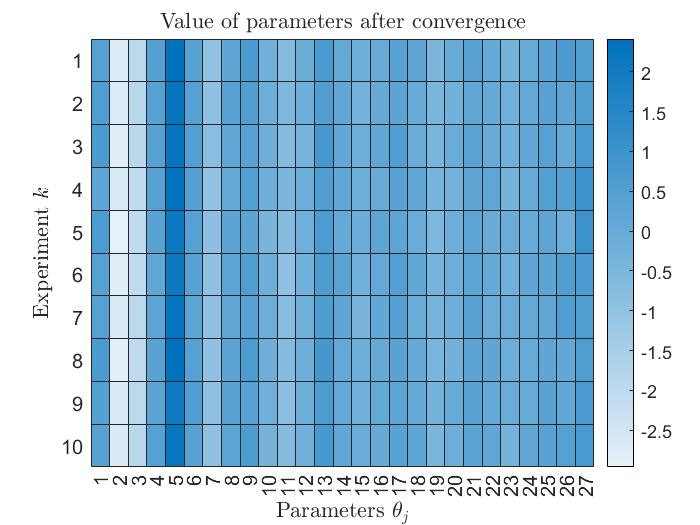}
}
\caption{\textbf{Example 2:} Heatmap of the parameters after convergence using different basis: (a)Chebyshev ($m=28$, $\alpha=0.01$) (b) Legendre ($m=28$, $\alpha=0.01$) (c) Fourier ($m=27$, $\alpha=0.01$).}
\label{fig:multiregionBasis}
\end{figure*}

\subsection{Example 3: Time-delay systems}
In this example, we use the same time delay system considered in \cite{Mi_SCC}.
Moreover, we fix the terminal state as $x(t_f)=0$.
\begin{align*}
\dot x(t) &= x(t)+x(t-1)+u(t), \quad 0\leq t \leq 2\\
x(t) &= 1,\quad -1 \leq t <0.
\end{align*}
The objective is to find the optimal control $u(t)$ to minimize the following cost functional:
\begin{align*}
J=\int_0^2 (x^2(t)+u^2(t)) \rm{d}t.
\end{align*}
The optimal solution is obtained using the augmented method mentioned in \cite{Mi-TD} by solving a fourth-order ordinary differential equation with the boundary condition to be $[\underline{x}(0)=1, \overline{x}(0)=\underline{x}(1), \overline{\lambda}(0)=\underline{\lambda}(1), \overline{x}(1)=0]$.
The optimal controller is thus $u(t)=-\lambda(t)$.
From Table \ref{tab:TDSys}, Fig. \ref{fig:TDSSys}, and Fig. \ref{fig:TDS_basis}, we can see the proposed algorithm obtains a cost close to the optimal solution and the terminal state is within the error tolerance.
The terminal state error can be further reduced by shrinking the tolerance value.
\begin{table}[!htp]
    \centering
    \begin{tabular}{|c|c|c|} \hline
        Algorithm &  Optimal Cost & Terminal state\\  \hline
         Chebyshev ($m=40$)& $ 6.7323\pm 0.0017 $  & $ -0.01\pm 0$ \\ \hline
         Legendre ($m=10$)&  $ 6.5891\pm 0.1089$ &  $ -0.0034\pm 0.0093$\\ \hline
         Fourier ($m=40$) & $ 6.6012 \pm 0.0100$ & $ 0.0000\pm 0.0104$ \\ \hline
        Optimal solution &  6.4880  & 0 \\ \hline
    \end{tabular} 
    \caption{Performance compare for Example 3.}
    \label{tab:TDSys}
    \vspace{-20pt}
\end{table}

\begin{figure*}[!htp]
\centering
\subfigure[]{
\includegraphics[width=0.25\textwidth]{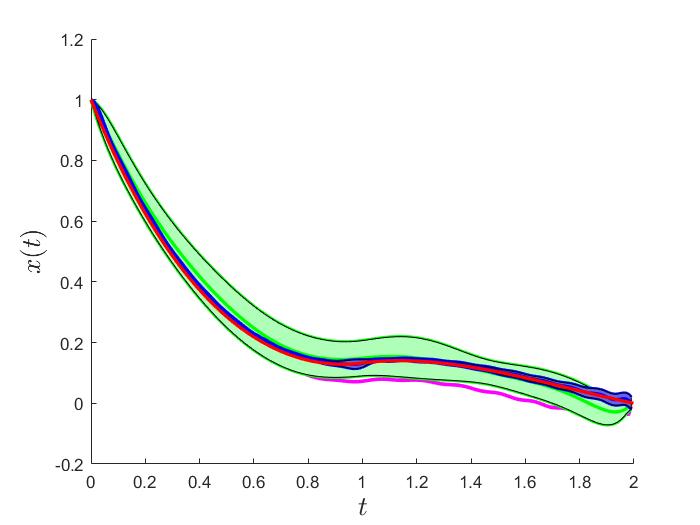}
}
\subfigure[]{
\includegraphics[width=0.25\textwidth]{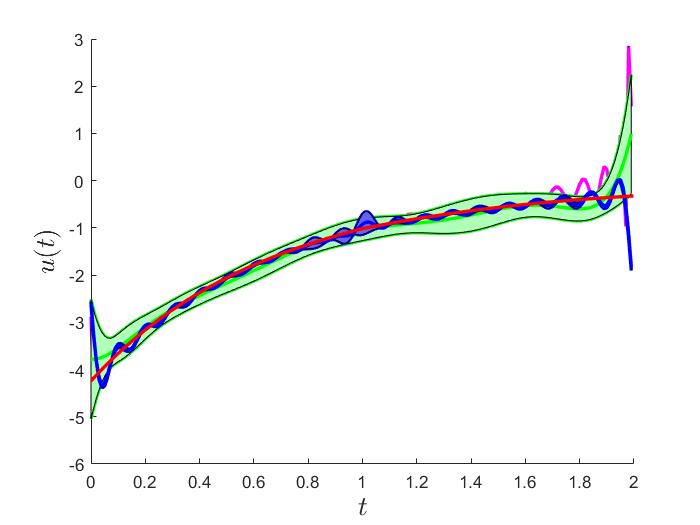}
}
\subfigure[]{
\includegraphics[width=0.25\textwidth]{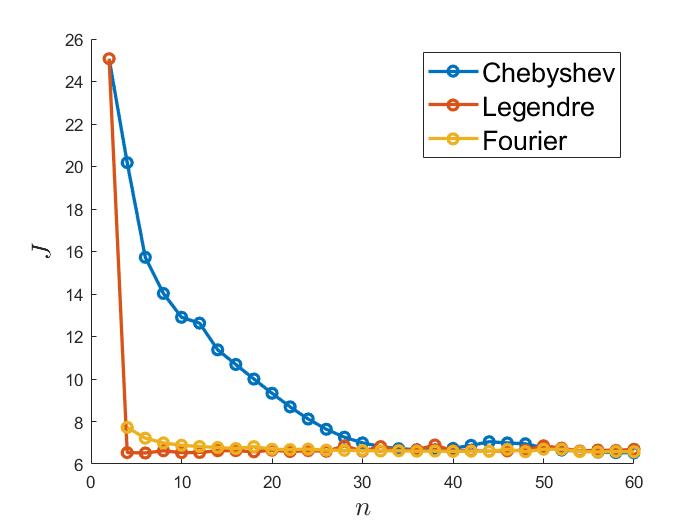}
}
\caption{\textbf{Example 3:} (a) state $x(t)$ under different basis functions (magenta: Chebyshev; green: Legendre; blue: Fourier); (b) control input $u(t)$ under different basis functions (magenta: Chebyshev; green: Legendre; blue: Fourier); (c) cost $J$ with respect to number of basis functions used.}
\label{fig:TDSSys}
\end{figure*}

\begin{figure*}[!htp]
\centering
\subfigure[]{
\includegraphics[width=0.25\textwidth]{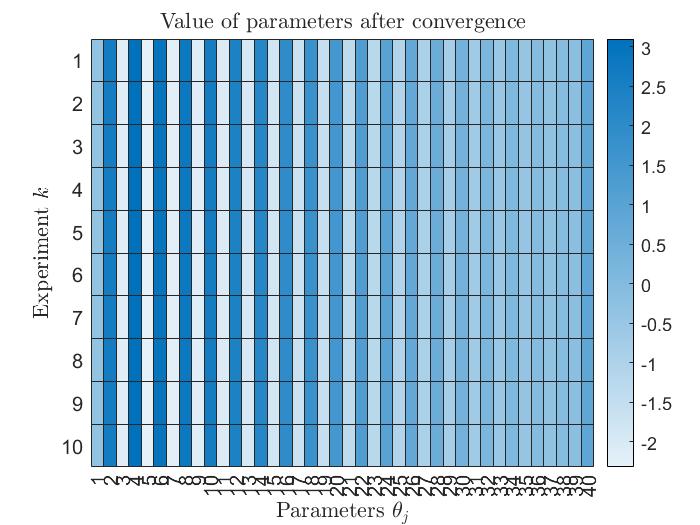}
}
\subfigure[]{
\includegraphics[width=0.25\textwidth]{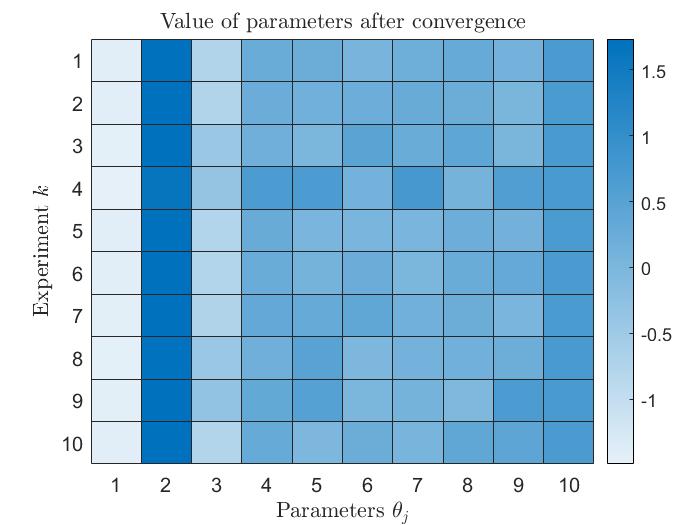}
}
\subfigure[]{
\includegraphics[width=0.25\textwidth]{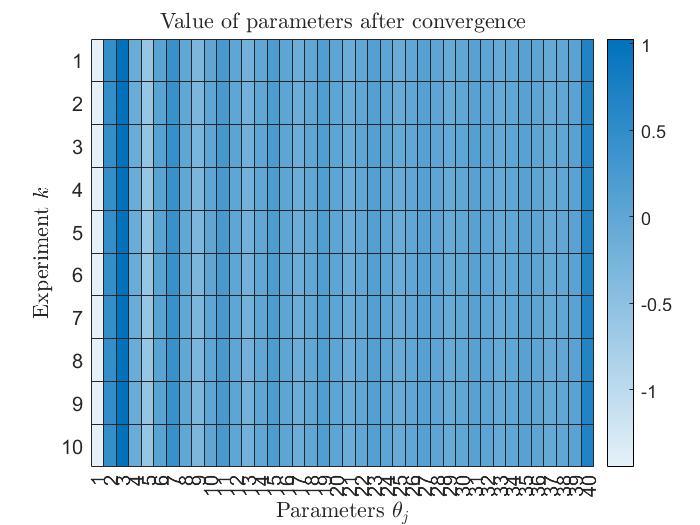}
}
\caption{\textbf{Example 3:} Heatmap of the parameters after convergence using different basis: (a)Chebyshev ($m=40$, $\alpha=0.01$) (b) Legendre ($m=10$, $\alpha=0.01$) (c) Fourier ($m=40$, $\alpha=0.01$).}
\label{fig:TDS_basis}
\end{figure*}

%%%%%%%%%%%%%%%%%%%%%%%%%%%%%%%%%%%%%%%
\section{Conclusions}\label{sec:conclusion}
In this article, we proposed a model-free optimal control method.
The control law and the Lagrange multipliers were parameterized gradients were obtained using least square methods.
The convergence of the proposed method was analyzed.
Examples including a continuous system, a switched system, and a time-delay system were provided to illustrate the algorithm.
One limit of this algorithm is that it relies heavily on tuning parameters to obtain convergence when the system has high dimensions.
We will be addressing these limitations in future research using neural networks as approximators.

\bibliographystyle{./IEEEtran} 
\bibliography{./IEEEabrv,./IEEEexample}

\end{document}